\documentclass[conference,letterpaper]{IEEEtran}

\usepackage{epsf,pgf,xcolor}
\usepackage{epstopdf,epsfig}
\usepackage[utf8]{inputenc}
\usepackage[english]{babel}
\usepackage{amsthm}

\usepackage{algorithm2e}
\usepackage{graphicx}% Include figure files
\graphicspath{ {C:/Users/alireza/Desktop/} }
\usepackage{dcolumn}% Align table columns on decimal point
\usepackage{bm}% bold math
\usepackage[export]{adjustbox}
\usepackage{array}
\usepackage{caption}
\usepackage[utf8]{inputenc} 
\usepackage[T1]{fontenc}
\usepackage{url}
\usepackage{epstopdf}
\usepackage{ifthen}
\usepackage[english]{babel}
\usepackage{amsthm}
\usepackage{caption}
\usepackage{cite}
\usepackage[english]{babel}
\usepackage{amssymb}
\usepackage[cmex10]{amsmath}
\newtheorem{theorem}{Theorem}
\newtheorem{lemma}{Lemma}

\graphicspath{ {./main/} } 
\usepackage{verbatim}
\providecommand{\keywords}[1]
{
  \textbf{\textit{Index Terms---}} #1
}
\begin{document}
\title{Age of Information in Multiple Sensing } 

% %%% Single author, or several authors with same affiliation:
% \author{%
%   \IEEEauthorblockN{Stefan M.~Moser}
%   \IEEEauthorblockA{ETH Zürich\\
%                     ISI (D-ITET)\\
%                     CH-8092 Zürich, Switzerland\\
%                     Email: moser@isi.ee.ethz.ch}
% }

%%% Several authors with up to three affiliations:
%\author{%
%  \IEEEauthorblockN{Alireza Javani}
%  \IEEEauthorblockA{Center for Pervasive Communications and Computing\\
% 					University of California, Irvine\\
%                  Email: ajavani@uci.edu}
%  \and
% \IEEEauthorblockN{Zhiying Wang }
% \IEEEauthorblockA{Center for Pervasive Communications and Computing\\
%  					University of California, Irvine\\             
%                    Email: zhiying@uci.edu}
%}

%%% Many authors with many affiliations:
 \author{%
   \IEEEauthorblockN{Alireza Javani,
                     Marwen Zorgui,
                     and Zhiying Wang
                     }
   \IEEEauthorblockA{
                     Center for Pervasive Communications and Computing\\
  					University of California, Irvine\\
                     \{ajavani, mzorgui, zhiying\}@uci.edu}

 }

\maketitle

%%%%%%
%% Abstract: 
%% If your paper is eligible for the student paper award, please add
%% the comment "THIS PAPER IS ELIGIBLE FOR THE STUDENT PAPER
%% AWARD." as a first line in the abstract. 
%% For the final version of the accepted paper, please do not forget
%% to remove this comment!
%%

\begin{abstract}
Having timely and fresh knowledge about the current state of information sources is critical in a variety of applications. In particular, a status update may arrive at the destination much later than its generation time due to processing and communication delays. The freshness of the status update at the destination is captured by the notion of age of information. In this study, we first analyze a network with a single source, $n$ servers, and the monitor (destination). The servers independently sense the source of information and send the status update to the monitor. We then extend our result to multiple independent sources of information in the presence of $n$ servers. We assume that updates arrive at the servers according to Poisson random processes. Each server sends its update to the monitor through a direct link, which is modeled as a queue. The service time to transmit an update is considered to be an exponential random variable. 
We examine both homogeneous and heterogeneous service and arrival rates for the single-source case, and only homogeneous arrival and service rates for the multiple sources case. 
We derive a closed-form expression for the average age of information under a last-come-first-serve (LCFS) queue for a single source and arbitrary $n$ homogeneous servers. For $n=2,3$, we derive the explicit average age of information for arbitrary sources and homogeneous servers, and for a single source and heterogeneous servers. For $n=2$ we find the optimal arrival rates given fixed sum arrival rate and service rates.
\end{abstract}

%% The paper must be self-contained. However, if you are referring to
%% a full version for checking certain proofs, please provide the
%% publically accessible location below.  If the paper is completely
%% self-contained, you can remove the following line from your
%% submission.
%%\textit{A full version of this paper is accessible at:}
%%\url{http://isit2019.fr/} 
\section{Introduction}
Widespread sensor network applications such as health monitoring using wireless sensors \cite{amin2018robust} and the Internet of things (IoT)\cite{chandana2018weather}, as well as applications like stock market trading and vehicular networks \cite{du2015effective}, require sending several status updates to their designated recipients (called monitors). Outdated information in the monitoring facility may lead to undesired situations. As a result, having the data at the monitor as fresh as possible is crucial. 

%In order to have a sense of the freshness of the received status update, the age of information
In order to quantify the freshness of the received status update, the age of information(AoI) metric was introduced in~\cite{kaul2012real}. For an update received by the monitor, AoI is defined as the time elapsed since the generation of the update. AoI captures the timeliness of status updates, which is different from other standard communication metrics like delay and throughput. It is affected by the inter-arrival time of updates and the delay that is caused by queuing during update processing and transmission.

In this paper, we consider AoI in a multiple-server network. We assume that a number of shared sources are sensed and then the data is transmitted to the monitor by $n$ independent servers. For example, the sources of information can be some shared environmental parameters, and independently operated sensors in the surrounding area obtain such information.
For another example, the source of information can be the prices of several stocks which is transmitted to the user by multiple independent service providers. Throughout this paper, a sensor or a service provider is called a server, since it is responsible to serve this update to the monitor.
We assume that status updates arrive at the servers independently according to Poisson random processes, and the server is modeled as a queue whose service time for an update is exponentially distributed. %We first consider the single source case, and then generalize to the model of multiple sources. 
We assume information sources are independent and are sensed by $n$ independent servers.

In \cite{kaul2012real}, authors considered the single-source single-server and first-come-first-serve (FCFS) queue model and determined the arrival rate that minimizes AoI. Different cases of multiple-source single-server under FCFS and last-come-first-serve (LCFS) were considered in \cite{yates2018age} and the region of feasible age was derived. In \cite{yates2018status, yates2018network}, the system is modeled as a source that submits status updates to a network of parallel and serial servers, respectively, for delivery to a monitor and AoI is evaluated. The parallel-server network is also studied in \cite{kam2016effect} when the number of servers is 2 or infinite, and the average AoI for FCFS queue model was derived.

%Minimizing AoI in a wireless camera network with correlated updates from the information source was considered in \cite{he2018minimizing}.
%Authors investigated different package managements in order to minimize AoI. 
Authors in \cite{kadota2016minimizing} formulated a discrete-time decision problem in order to find a scheduling policy for minimizing the expected weighted sum of AoI. A multi-source multi-hop setting in broadcast wireless networks was investigated in \cite{farazi2018age} and a fundamental lower bound on  the average AoI was derived.
Different scheduling policies with throughput constraints were considered in \cite{kadota2018optimizing} to minimize AoI. Another age-related metric of peak AoI was introduced in \cite{costa2016age}, which corresponds to the age of information at the monitor right before the receipt of the next update. 
The average peak AoI minimization in IoT networks and wireless systems was considered in \cite{abd2018average, he2016optimal}. The problem of minimizing the average age in energy harvesting sources by manipulating the update generation process was studied in \cite{wu2018optimal, feng2018minimizing}. Maximizing energy efficiency of wireless sensor networks that include constraints on AoI is investigated in \cite{valehi2017maximizing}. 

In this paper, we study the average age of information as in~\cite{kaul2012real}. We mainly consider LCFS with preemption in service  (in short, LCFS) queue model, namely, upon the arrival of a new update, the server immediately starts to serve it and drops any old update being served. We derive a closed-form formula of the average AoI for LCFS and a single source. For multiple sources, AoI formula is derived for arbitrary number of sources and $n=2,3$ servers. In addition, the heterogeneous network with a single source is considered. To obtain the AoI, we use the stochastic hybrid system (SHS) analysis similar to \cite{yates2018status, yates2018age}.  %Moreover, LCFS, LCFS with preemption in waiting, and FCFS queue models are compared in terms of the average AoI through simulation.
%We also provide optimal arrival rate for a given service rate in the FCFS case numerically. 
% For multiple source, recursive algorithms are shown to calculate AoI.

This paper is organized as follows. Section \ref{sec:preliminaries} formally introduces the system model of interest, and provides preliminaries on SHS. In subsection \ref{LCFS}, we derive the average age of information formula by applying SHS method to our model when we have a signle information source and the network is homogeneous. In subsection \ref{multiple}  we derive AoI for arbitrary number of information sources when $n=2,3$. In section \ref{hetro-sec}, we investigate the heterogeneous network when we have a single source and $n=2,3$ and find the optimal arrival rate at each server when $n=2$. At the end, the conclusion follows in section \ref{conc}.
\section{System Model and Preliminaries}
\label{sec:preliminaries}
Notation: in this paper, we use boldface for vectors, and normal font with a subscript for its elements. For example, for a vector $\mathbf{x}$, the $j$-th element is denoted by $x_j$. For non-negative integers $a$ and $b \geq a$, we define $[a:b] \triangleq  \{a,  \ldots, b\}$, $[a]\triangleq  [1:a]$. If $a > b$, $[a:b] = \emptyset$.

In this section, we first present our network model, and then briefly review the stochastic hybrid system analysis from \cite{yates2018age}.
The network consists of $m$ information sources that are sensed by $n$ independent servers as illustrated in Figure \ref{fig2}. Updates after going through separate links are aggregated at the monitor side. The interest of this paper is the average AoI at the monitor. Server $j$ collects updates of source $i$ following a Poisson random process with rate $\lambda_{j}^{(i)}$ and the service time is an exponential random variable with average $\frac{1}{\mu_{j}}$, independent of all other servers, $j \in [n], i \in [m]$. A network is called homogeneous if $\lambda_{j}^{(i)}=\lambda^{({i})}, \mu_j=\mu$, for all $j \in [n], i \in [m]$, otherwise, it is heterogeneous. In case of a single source in a homogeneous network, we denote $\lambda^{({1})}$ simply by $\lambda$. 
%After deriving the explicit formula of AoI for a single source, using similar method we introduce an algorithm that calculates AoI when we have several independent information sources. Similarly, we assume that $\lambda_i=\lambda, \mu_i=\mu$ for all $1 \le i \le n$ in our analysis. At the end, we consider heterogeneous arrival and service rates when we have a single source and $3$ servers and the general case with any number of servers will be investigated in our future works.

Consider a particular source. Suppose the freshest update at the monitor at time $t$ is generated at time $u(t)$, the \emph{age of information} at the monitor (in short, AoI) is defined as $\Delta(t) =t-u(t)$, which is the time elapsed since the generation of the last received update. 
%Our goal is compute the average AoI, which is $\Delta \triangleq \lim_{T \to \infty}   \int_0^T \Delta(t) \big/ T$.
%Suppose the freshest update at the monitor at time $t$ is generated at time $u(t)$, the AoI at the monitor is defined as $\Delta(t) =t-u(t)$, which is the time elapsed since the generation of the last received update. 
From the definition, it is clear that AoI linearly increases at a unit rate with respect to $t$, except some reset jumps to a lower value at points when the monitor receives a fresher update from the source. The age of information of our network is shown in Figure \ref{fig1}.
Let $t_{1}, t_{2},\dots, t_N$ be the generation time of all updates at all servers in increasing order. 
The black dashed lines show the age of every update.
Let $T_{1}, T_{2},\dots, T_N$ be the receipt time of all updates. The red solid lines show AoI.

We note a key difference between the model in this work and most previous models. Updates come from different servers, therefore they might be out of order at the monitor and thus a new arrived update might not  have any effect on AoI because a fresher update is already delivered. As an example, from the $6$ updates shown in Figure \ref{fig1}, \emph{useful} updates that change AoI are updates $1,3,4$ and $6$, while the rest are disregarded as their information when arrived at the monitor is obsolete. 
Thus among all the received updates for AoI analyses, we only need to consider the \emph{useful} ones that lead to a change in AoI. 
%We note a key difference between the model in this work and most previous models: because the updates come from different servers they might be out of order at the monitor and thus a newly arrived update might not have any effect on the AoI.
\begin{figure}
\centering
\includegraphics[width=0.42\textwidth]{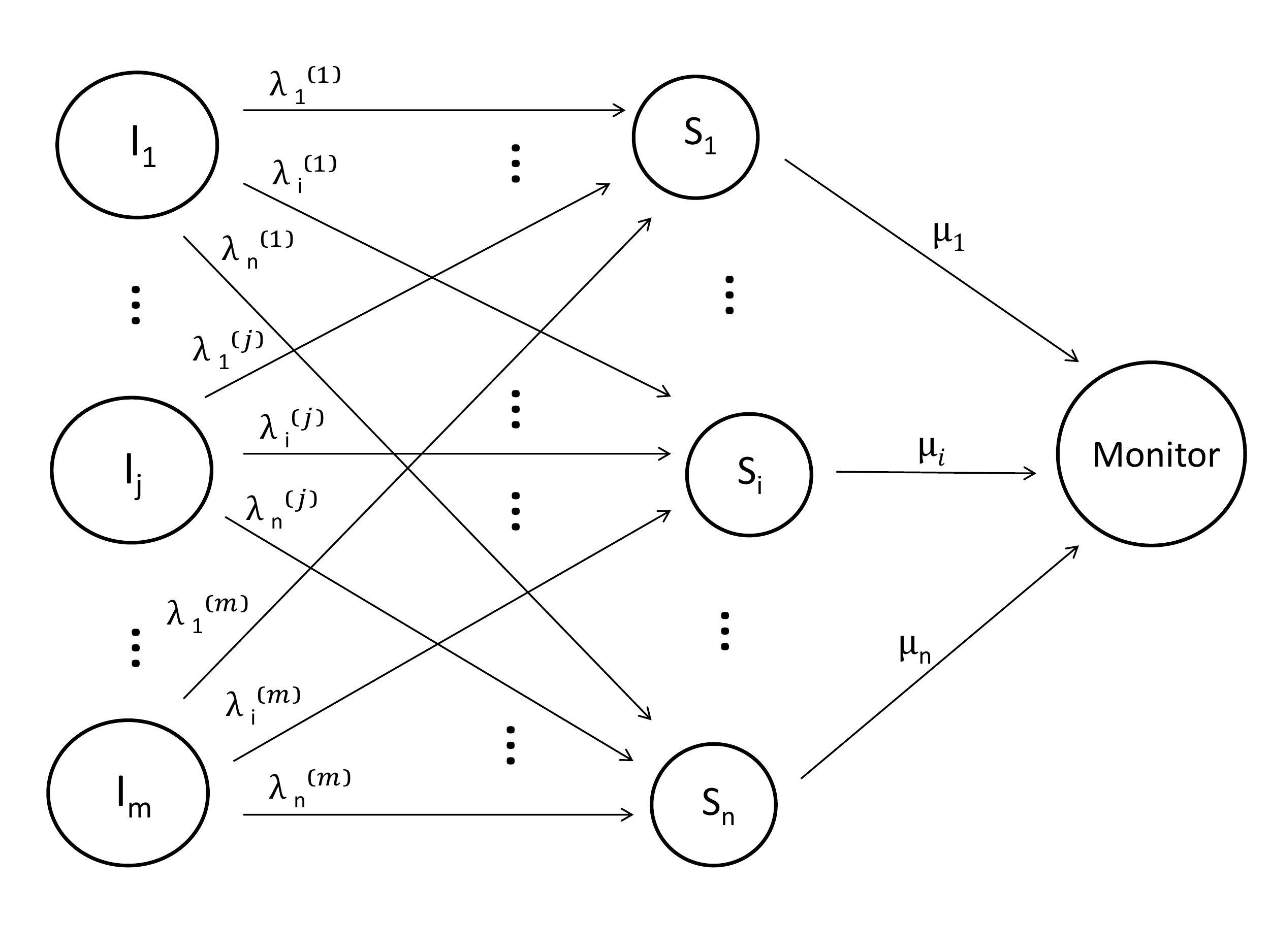} 
\caption{The $n$-server monitoring network with $S_{1},S_{2},...,S_{n}$ being the servers and $I_{1},I_{2},...,I_m$ being the independent information sources, sending the updates from the sources to the monitor.}

\label{fig2}
\end{figure}

\begin{figure}
\centering
\includegraphics[width=0.4\textwidth]{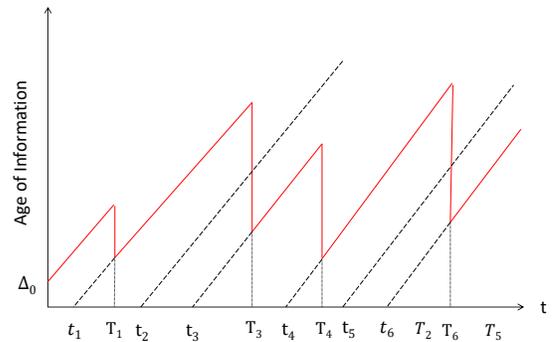} 
\caption{AoI for a network with $n$ servers.}
\label{fig1}
\end{figure}

The average AoI is the limit of the average age over time $\Delta \triangleq \lim_{T \to \infty}   \int_0^T \Delta(t) \big/ T$, and for a stationary ergodic system, it is also the limit of the average age over the ensemble $\Delta= \lim_{t \to \infty} \mathbb{E}[\Delta(t)]$.
\begin{comment}
Based on \cite[Thm. 1]{yates2018age} for a stationary ergodic status updating system we have:
\begin{equation}\label{eq1}
\Delta= \frac{\mathbb{E}{[YT]}+ \mathbb{E}{[Y^2]}/2}{\mathbb{E}{[Y]}},
\end{equation}
where $Y$ is the inter-arrival time of useful updates at the monitor and $T$ is the time each useful update spends in the system. Calculating closed-form expressions of \eqref{eq1} is not tractable in general.
\end{comment}

In the paper, we view our system as a stochastic hybrid system (SHS) and apply a method first introduced in \cite{yates2018age} in order to calculate AoI. 
We can thus obtain the average AoI under LCFS with preemption in service, or in short, LCFS.

In SHS, the state is composed of a discrete state and a continuous state. 
The discrete state $q(t) \in \mathcal{Q}$, for a discrete set $\mathcal{Q}$, is a continuous-time discrete Markov chain (e.g., to represent the number of idle servers in the network), and the continuous-time continuous state $\mathbf{x}(t) = [x_0(t),x_1(t),\dots,x_n(t)] \in \mathbb{R}^{n+1}$ is the stochastic process for AoI. We use $x_0(t)$ to represent the age at the monitor, and $x_j(t)$ for the age at the $j$-th server, $j=1,2,\dots,n$. Graphically, we represent each state $q \in \mathcal{Q}$ by a node. For the discrete Markov chain $q(t)$, transitions happen from one state to another through directed transition edge $l$, and the time spent before the transition occurs is exponentially distributed with rate $\lambda{(l)}$. Note that it is possible to transit from the same state to itself. The transition occurs when an update arrives at a server, or an update is received at the monitor. Thus the transition rate is the update arrival rate or the service rate  $\lambda{(l)} \in [\lambda_{1}^{(1)},...,\lambda_{n}^{(m)}, \mu_{1},...,\mu_{n}]$. 
Denoted by $L'_{{q}}$ and $L_{{q}}$ the sets of incoming and outgoing transitions of state $q$, respectively.
When transition $l$ occurs, we write that the discrete state transits from $q_l$ to $q_l^\prime$. For instance, if we have $2$ states and considering the transition $l$ from state $1$ to state $2$, we have $q_{l}=1$ and $q_{l}^{\prime} = 2$ which shows that state $2$ is an outgoing transition for state $1$ and state $1$ is an incoming transition for state $2$. For a transition, we denote that
the continuous state changes from $\mathbf{x}$ to $\mathbf{x}^\prime$.
In our problem, this transition is linear in the vector space of $\mathbb{R}^{n+1}$, i.e., $\mathbf{x}^\prime=\mathbf{x}A_{l}$, for some real matrix $A_{l}$ of size $(n+1) \times (n+1)$.  
Note that when we have no transition, the age grows at a unit rate for the monitor and relevant servers, and is kept unchanged for irrelevant servers.
Hence, within the discrete state $q$, $\mathbf{x}(t)$ evolves as a piece-wise linear function in time, namely, $\frac{\partial{\mathbf{x}(t)}}{\partial{t}} = \mathbf{b}_{q}$,  for some $ \mathbf{b}_q \in \{0,1\}^{n+1}$. In other words,  the age grows at a unit rate for the monitor and relevant servers; and the age is kept unchanged for irrelevant servers.
For our purpose, we consider the discrete state probability
\begin{equation}
\pi_{\hat{q}} (t) \triangleq \mathop{\mathbb{E}}[\delta_{\hat{q},q(t)}] =P[q(t)=\hat{q}],
\end{equation}
and the correlation between the continuous state $\mathbf{x}(t)$ and the discrete state $q(t)$:
\begin{equation}
\mathbf{v}_{\hat{q}} = [v_{\hat{q}0} (t),\dots, v_{\hat{q}n} (t)] \triangleq \mathop{\mathbb{E}}[\mathbf{x}(t) \delta_{\hat{q},q(t)}].
\end{equation}
Here $\delta_{\cdot,\cdot}$ denotes the Kronecker delta function. When the discrete state $q(t)$ is ergodic, $\mathbf{\pi}_q(t)$ converges uniquely to the stationary probability ${\mathbf{\pi}}_q$, for all $q \in \mathcal{Q}$. We can find these stationary probabilities from the following set of equations knowing that $\sum_{q \in \mathcal{Q}}^{}\pi_{q} = 1$,
\begin{align*}
{\mathbf{\pi}}_{{q}} \sum_{l \in L_{{q}}}^{} \lambda{(l)}=
 \sum_{l \in L^\prime_{{q}}}^{}
\lambda{(l)} {\mathbf{\pi}}_{q_l}. \quad {q} \in \mathcal{Q}
\end{align*}

A key lemma we use to develop AoI for our LCFS queue model is the following from \cite{yates2018age}, which was derived from the general SHS results in \cite{hespanha2006modelling}.

\begin{lemma} \label{lem:yates}
\cite{yates2018age}
 If the discrete-state Markov chain $q(t)$ is ergodic with stationary distribution ${\pi}$ and we can find a non-negative solution of $ \{{\mathbf{v}}_{{q}}, {q} \in \mathcal{Q} \}$ such that 
\begin{equation}\label{eq:yateslemma}
{\mathbf{v}}_{{q}} \sum_{l \in L_{{q}}}^{} \lambda{(l)}=
\mathbf{b}_{{q}} {\pi}_{{q}} + \sum_{l \in L^\prime_{{q}}}^{}
\lambda{(l)} {\mathbf{v}}_{q_l} A_{l}, \quad {q} \in \mathcal{Q},
\end{equation}
then the average age of information is given by 
\begin{align}
\Delta= \sum_{{q} \in \mathcal{Q}}^{} {v}_{{q} {0}}.
\end{align} 
\end{lemma}

\section{AoI in Homogeneous Networks}
\subsection{Single Source Multiple Sensors}
\label{LCFS}
In this section, we present AoI calculation with the LCFS queue for the single-source $n$-server homogeneous network. In this network, upon arrival of a new update, each server immediately drops any previous update in service and starts to serve the new update.
Note that to compute the average AoI, Lemma \ref{lem:yates} requires solving $|\mathcal{Q}|(n+1)$ linear equations of $ \{{\mathbf{v}}_{{q}}, {q} \in \mathcal{Q} \}$. To obtain explicit solutions for these equations, the complexity grows with the number of discrete states. Since the discrete state typically represents the number of idle servers in the system for homogeneous servers, $|\mathcal{Q}|$ should be $n+1$. In the following, we introduce a method inspired by \cite{yates2018status} to reduce the number of discrete states and efficiently describe the transitions.

We define our continuous state $\mathbf{x}$ at a time as follows:
the first element of $\mathbf{x}$ is AoI at the monitor ($x_{0}$), the second is always the freshest update among all updates in the servers, the third is always the second freshest update in the servers, etc. With this definition we always have $x_{1} \leq x_{2} \leq .... \leq x_{n}$, for any time. Note that the index $i$ of $x_i$ does not represent a physical server index, but the $i$-th smallest age of information among the $n$ servers. The physical server index for $x_i$ changes with each transition. We say that the server corresponding to $x_i$ is the $i$-th \emph{virtual} server. 

A transition $l$ is triggered by (i) the arrival of an update at a server, or (ii) the delivery of an update to the monitor. Recall that we use $\mathbf{x}$ and $\mathbf{x}'$ to denote AoI continuous state vector right before and after the transition $l$.

When one update arrives at the monitor and the server for that update becomes idle, we put a \emph{fake update} to the server using the method introduced in \cite{yates2018status}. 
Thus we can reduce the calculation complexities and only have one discrete state indicating that all servers are virtually busy. We denote this state by $q=0$.
In particular, we put the current update that is in the monitor to an idle server until the next update reaches this server. This assumption does not affect our final calculation for AoI, because even if the fake update is delivered to the monitor, AoI at the monitor does not change.

When an update is delivered to the monitor from the $k$-th virtual server, the server becomes idle and as previously stated, receives the fake update. The age at the monitor becomes $x'_0=x_k$, and the age at the $k$-th server becomes $x'_k = x'_0=x_k$.
In this scenario, consider the update at the $j$-th virtual server, for $j>k$.  Its delivery to the monitor does not affect AoI since  it is older than the current update of the monitor, i.e., $x_j \ge x_k = x'_0$. 
Hence, we can adopt a \emph{fake preemption} where the update for the $j$-th virtual server, for all $k \le j \le n$, is preempted and replaced with the fake current update at the monitor.
Physically, these updates are not preempted and as a benefit, the servers do not need to cooperate and can work in a distributed manner. 

\begin{figure}
\centering
\includegraphics[width=0.2\textwidth]{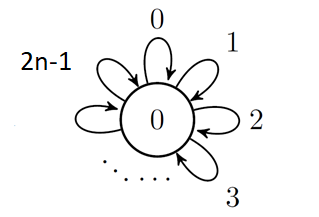} 
    \caption{SHS for our model with $n$ servers.}
   
  \label{fig3}  
\end{figure}

\begin{table}
%\label{table1}
\centering
\begin{tabular}{ cccccccccc }
 
 $l$ & $\lambda{(l)}$ & $\mathbf{x}^\prime$ =$\mathbf{x}A_{l}$ \\          \hline
 $0$ & $\lambda$ & $[x_{0},0,x_{2},x_{3},x_{4},...,x_{n}]$ \\ \hline
 $1$ & $\lambda$ & $[x_{0},0,x_{1},x_{3},x_{4},...,x_{n}]$ \\ \hline
 $2$ & $\lambda$ & $[x_{0},0,x_{1},x_{2},x_{4},...,x_{n}]$ \\ \hline
 
  & $\vdots$ & $\vdots$ \\  \hline
 $n-1$ & $\lambda$ & $[x_{0},0,x_{1},x_{2},x_{3},..,x_{n-1}]$ \\ \hline
 $n$ & $\mu$ & $[x_{1},x_{1},x_{1},x_{1},...,x_{1}]$ \\  \hline
 $n+1$ & $\mu$ & $[x_{2},x_{1},x_{2},x_{2},...,x_{2}]$ \\  \hline
 $n+2$ & $\mu$ & $[x_{3},x_{1},x_{2},x_{3},...,x_{3}]$ \\  \hline
  & $\vdots$ & $\vdots$ \\  \hline
 $2n-1$ & $\mu$ & $[x_{n},x_{1},x_{2},x_{3},...,x_{n}]$ \\  
 
\end{tabular}
\caption{Table of transformation for the Markov chain in Figure~\ref{fig3}.}
\label{table:table 1}
\end{table}
%By utilizing the fake update and fake preemption, i.e, putting the current age of the monitor to the idle servers and preempting all the updates older than that of the monitor, 
By utilizing virtual servers, fake update, and fake preemption, we reduce SHS to a single discrete state with linear transition $A_l$.
We illustrate our SHS with discrete state space of $Q=\{0\}$ in Figure \ref{fig3}. The stationary distribution ${\pi}_{0}$ is trivial and ${\pi}_{0}=1$. We set $\mathbf{b}_{q}=[1,...,1]$ which indicates that the age at the monitor and the age of each update in the system grows at a unit rate. The transitions are labeled $l  \in \{0,1,...,2n-1\}$ and for each transition $l$ we list the transition rate and the transition mapping in Table \ref{table:table 1}. For simplicity, we drop the index $q=0$ in the vector $\mathbf{v}_0$, and write it as $\mathbf{v}=[v_0,v_1,\dots,v_n]$. Because we have one state, $\mathbf{x}A_{l}$ and $\mathbf{v}A_{l}$ are in correspondence. Next, we describe the transitions in Table \ref{table:table 1}. 

{\bf Case I.} $l=0,1,..,n-1:$ When a fresh update arrives at virtual server $l+1$, the age at the monitor remains the same and $x_{l+1}$ becomes zero. This server has the smallest age, so we take this zero and reassign it to the first virtual server, namely, $x^\prime_{1}=0$. In fact virtual servers $1,2,\dots,l+1$ all get reassigned virtual server numbers. Specifically, after transition $l$, virtual server $l+1$ becomes virtual server $1$, and virtual server $1$ becomes virtual server $2$,..., virtual server $l$ becomes virtual server $l+1$. The transition rate is the arrival rate of the update, $\lambda$. The matrix $A_l$ is

\vspace{-0.35cm}
\begin{footnotesize}
\begin{align}
    \bordermatrix{
~	&	0	&	1	&	2	&	\dots	&	l+1	&	l+2	&	\dots	&	n	\cr
0	&	1	&		&		&		&		&		&		&		\cr
1	&		&	0	&	1	&		&		&		&		&		\cr
\vdots	&		&		&		&	\ddots	&		&		&		&		\cr
l	&		&		&		&		&	1	&		&		&		\cr
l+1	&		&		&		&		&		&	0	&		&		\cr
l+2	&		&		&		&		&		&	1	&		&		\cr
\vdots	&		&		&		&		&		&		&	\ddots	&		\cr
n	&		&		&		&		&		&		&		&	1	\cr
    }.
\end{align}
\end{footnotesize}

\vspace{-0.35cm}
\noindent{\bf Case II.} $l=n,n+1,..,2n-1:$ When an update is received at the monitor from virtual server $l+1-n$, the age at the monitor changes to $x_{l+1-n}$ and this server becomes idle. Using fake updates and fake preemption we assign $x'_{j}=x_{l+1-n}$, for all $l+1-n \le j \le n$. The transition rate is the service rate of a server, $\mu$. The matrix $A_l$ is

\vspace{-0.25cm}
\begin{footnotesize}
\begin{align}
\bordermatrix{
~	&	0	&	1	&			\dots	&	l-n	&	l+1-n	&	\dots	&	n	\cr
0	&	0	&		&				&		&		&		&		\cr
1	&		&	1	&				&		&		&		&		\cr
\vdots	&		&		&			\ddots	&		&		&		&		\cr
l-n	&		&		&				&	1	&		&		&		\cr
l+1-n	&	1	&	0	&			\dots	&	0	&	1	&	\dots	&	1	\cr
l+2-n	&	0	&		&			\dots	&	\dots	&	\dots	&	\dots	&	0	\cr
\vdots	&	\vdots	&		&				&		&		&		&	\vdots	\cr
n	&	0	&		&			\dots	&	\dots	&	\dots	&	\dots	&	0	\cr
}.
\end{align} 
\end{footnotesize}

\vspace{-0.25cm}
\noindent	Below we state our main theorem on the average AoI for the single-source $n$-server network.

\begin{theorem} \label{theory}
The age of information at the monitor for homogeneous single-source $n$-server network where each server has a LCFS queue is:

\vspace{-.3cm}
\begin{small}
\begin{align}
\label{AoI_single_source}
&AoI 
%\sum_{j=2}^{n} w_{j} +\frac{1}{n\lambda} +\frac{\lambda}{n\mu} w_{n}= \\
%&\sum\limits_{j = 2}^n \frac{1}{n \lambda} \prod\limits_{i=1}^{j-1} \frac{\rho (n-i+1)}{i    + (n-i) \rho} + 
% \frac{1}{n \lambda}
%+
% \frac{1}{n \mu}
%\frac{1}{n  } \prod\limits_{i=1}^{n-1} \frac{\rho (n-i+1)}{i    + (n-i) \rho} \\
 =
  \frac{1}{\mu} \left[ 
\frac{1}{n \rho} \sum\limits_{j = 1}^{n-1}  \prod\limits_{i=1}^{j } \frac{\rho (n-i+1)}{i    + (n-i) \rho}+ 
 \frac{1}{n \rho}
+
\frac{1}{n^2  } \prod\limits_{i=1}^{n-1} \frac{\rho (n-i+1)}{i    + (n-i) \rho} 
\right],
\end{align}
\end{small}

\vspace{-.3cm}
\noindent where $\rho = \frac{\lambda}{\mu}$.
%where, %$w_2=\frac{1}{(n-1)\lambda+\mu}$, $2 \le j \le n-1$, and
%\begin{align*}
%w_j &= \frac{1}{n \lambda} \prod\limits_{i=1}^{j-1} \frac{\lambda (n-i+1)}{i \mu  + (n-i) \lambda}\\
%&= \frac{1}{n \lambda} \prod\limits_{i=1}^{j-1} \frac{\rho (n-i+1)}{i    + (n-i) \rho}
%\quad 2 \le j \le n.
%\end{align*}
\end{theorem}
\begin{proof}
Recall that $\mathbf{v}$ denotes the vector  $\mathbf{v}_0$ for the single state $q=0$.
By Lemma \ref{lem:yates} and the fact that there is only one state,
we need to calculate the vector $\mathbf{v}$ as a solution to \eqref{eq:yateslemma}, and the $0$-th coordinate $v_{0}$ is AoI at the monitor. As we mentioned $\mathbf{v}A_{l}$ is in correspondence with $\mathbf{x}A_{l}$, so we have:

\vspace{-.4cm}
\begin{small}
\begin{align}
 (n\lambda+n\mu){\mathbf{v}}= & \quad \quad [1,1,1,1,1,1,1,...,1] \nonumber \\
   &+\lambda[v_{0},0,v_{2},v_{3},v_{4},...,v_{n}]  \nonumber \\ 
   &+\lambda[v_{0},0,v_{1},v_{3},v_{4},...,v_{n}]  \nonumber \\
   &+\lambda[v_{0},0,v_{1},v_{2},v_{4},...,v_{n}]  \nonumber \\
   & \quad \quad \quad \quad \vdots \quad \quad \quad \quad    \quad \vdots  \nonumber \\
   &+\lambda[v_{0},0,v_{1},v_{2},v_{3},...,v_{n-1}]  \nonumber \\
   &+\mu[v_{1},v_{1},v_{1},v_{1},v_{1},...,v_{1}]   \nonumber \\
   &+\mu[v_{2},v_{1},v_{2},v_{2},v_{2},...,v_{2}]   \nonumber \\
   &+\mu[v_{3},v_{1},v_{2},v_{3},v_{3},...,v_{3}]   \nonumber \\
   & \quad \quad \quad \quad \vdots \quad \quad \quad \quad    \quad \vdots  \nonumber \\
   &+\mu[v_{n},v_{1},v_{2},v_{3},...,v_{n-1},v_{n}]   \label{eq:main}.
\end{align}  
\end{small}

\vspace{-.3cm}
\noindent From the $0$-th coordinate of \eqref{eq:main}, we have $(n\lambda+n\mu)v_{0}= 1+n\lambda v_{0} + \mu \sum_{j=1}^{n} v_{j}$, implying
\begin{align}
%(n\lambda+n\mu)v_{0}= 1+n\lambda v_{0} + \mu \sum_{j=1}^{n} v_{j},   \implies 
  v_{0}= \frac{1}{n\mu} + \frac{\sum_{j=1}^{n} v_{j}}{n} \label{eq3}.
\end{align}
From the $1$-st coordinate of \eqref{eq:main}, it follows  that 
%\begin{align}
$v_{1}~=~\frac{1}{n\lambda}$.
%\end{align}
Then, to calculate $v_{0}$, we have to calculate $v_{i}$ for $i \in \{2,...,n\}$. From the $i$-th coordinate of \eqref{eq:main},

\vspace{-0.35cm}
\begin{small}
\begin{align}
%& ((n-i+1)\lambda+(i-1)\mu) v_{i} \nonumber\\
%= & 1+ \mu \sum_{j=1}^{i-1} v_{j} +\lambda (n-i+1) v_{i-1}.
((n-i+1)\lambda+(i-1)\mu) v_{i} =   1+ \mu \sum_{j=1}^{i-1} v_{j} +\lambda (n-i+1) v_{i-1}.
\label{w2}
\end{align}
\end{small}

\vspace{-0.4cm}
\noindent	For $i \in \{2,3,...,n-1\} $, from \eqref{w2}, we obtain
\begin{align*}
(i\mu +(n-i)\lambda)(v_{i+1}-v_{i}) = \lambda (n-i+1) (v_{i}-v_{i-1}) .
\end{align*} 
Hence,
%\begin{align*}
$w_{i+1} \triangleq v_{i+1}-v_{i} 
%&=\frac{\lambda (n-i+1)}{(i\mu +(n-i)\lambda)} (v_{i}-v_{i-1}) 
%&\nonumber\\
=\frac{\lambda (n-i+1)}{(i\mu +(n-i)\lambda)} w_i$.
%\end{align*}
\noindent Setting $i=2$ in \eqref{w2}, we have
\begin{align}
    ((n-1)\lambda+\mu)v_{2}=1+ \mu v_{1} +\lambda (n-1)v_{1}.
    \label{w_2_source}
\end{align}
Simplifying \eqref{w_2_source}, we obtain 
%$w_{2}=v_{2}-v_{1}=\frac{2-1/n+\frac{\mu}{n\lambda}}{\mu+(n-1)\lambda} - \frac{1}{n\lambda}=\frac{1}{(n-1)\lambda+\mu}$.
 $w_{2}=v_{2}-v_{1}=\frac{1}{(n-1)\lambda+\mu}$. Therefore, we write
 \begin{align}
w_j &= \frac{1}{n \lambda} \prod\limits_{i=1}^{j-1} \frac{\lambda (n-i+1)}{i \mu  + (n-i) \lambda}, 2 \le j \le n.
\label{w_j}
 \end{align}
%As a result by replacing $v_{1}$ with $\frac{1}{n\lambda}$, we get $v_{2}= \frac{2-1/n+\frac{\mu}{n\lambda}}{\mu+(n-1)\lambda}$.
%Therefore $w_{2}=v_{2}-v_{1}=\frac{2-1/n+\frac{\mu}{n\lambda}}{\mu+(n-1)\lambda} - \frac{1}{n\lambda}=\frac{1}{(n-1)\lambda+\mu}.$  
Finally, setting $i=n$ in \eqref{w2}, 
\begin{align}
   (\lambda+ (n-1)\mu)v_{n}=  1+ \mu \sum_{j=1}^{n-1} v_{j}+ \lambda v_{n-1} ,
\end{align}
implying
 $   \mu  \sum_{i=1}^{n} v_{i}
 =   \mu \sum_{j=1}^{n-1} v_{j}+\mu v_{n}  
 =   (\lambda+ (n-1)\mu)v_{n} + \mu v_{n} -1 - \lambda v_{n-1}.
$
Hence,
\begin{align}\label{eq2}
 \frac{1}{n} \sum_{i=1}^{n} v_{i} = \frac{\lambda}{n\mu} w_{n} +v_{n}-\frac{1}{n\mu}.
\end{align}
Combining \eqref{eq3} and \eqref{eq2}, we obtain the average AoI as
\begin{align*}
AoI = v_0=v_{n} +\frac{\lambda}{n\mu} w_{n} =\sum_{j=2}^{n} w_{j} +\frac{1}{n\lambda} +\frac{\lambda}{n\mu} w_{n},
\end{align*}
which is simplified to \eqref{AoI_single_source} using \eqref{w_j}.
%Thus the proof is completed.
\end{proof}
Figure \ref{fig4} shows AoI when the total arrival rate $n\lambda$ is fixed and $n=1,2,3,4,10$. We observe that for up to $4$ servers,  a significant decrease in AoI occurs with the increase of $n$. However, increasing the number of servers beyond $4$ provides only a negligible decrease in AoI. %and it is therefore inefficient from a network-resource point of view. 
\begin{figure}
\centering
\includegraphics[width = 0.43\textwidth]{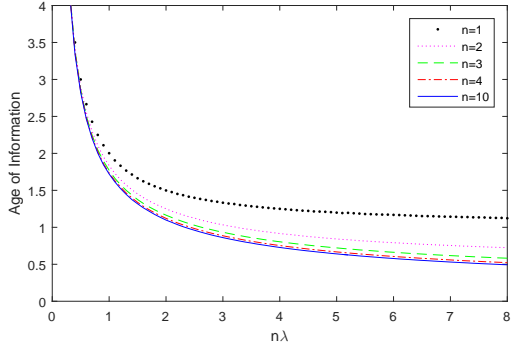} 
\caption{AoI versus the number of servers, for fixed total arrival rate. For each server, the service rate $\mu=1$ and the total arrival rate $n\lambda$ is shown in the x-axis.}
\label{fig4}
\end{figure}
In Figure \ref{fig6}, LCFS (with preemption in service), LCFS with preemption in waiting, and FCFS queue models are compared numerically. As can be seen from the figure, LCFS outperforms the other two queue  models, which coincides with the intuition that exponential service time is memoryless and older updates in service should be preempted. Moreover, we observe that the optimal arrival rate for FCFS queue is approximately $0.5$ for all $n \le 50$.
\begin{figure}
\centering
\includegraphics[width=0.45\textwidth]{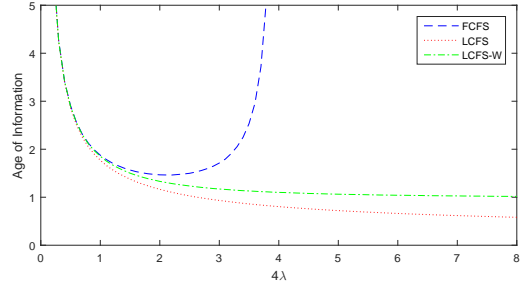}
\caption{Comparison of LCFS, FCFS, and LCFS with preemption in waiting (LCFS-W). The number of servers is $n=4$ and $\mu=1$ for each server.}
\label{fig6}
\end{figure}
%\begin{remark}
%When $n=2,3$ assuming $\rho =\frac{\lambda}{\mu}$, AoI reduces to $\frac{1}{2\mu}(1+\frac{1}{\rho}+\frac{1}{1+\rho})$ and $\frac{1}{3\mu}(1+\frac{1}{\rho}+\frac{4(1+\rho)}{(2+\rho)(2\rho+1)})$, respectively.
%\end{remark}
\subsection{Multiple Sources Multiple Sensors}
\label{multiple}
In this subsection, we present AoI calculation with the LCFS queue for the $m$-source $n$-server homogeneous network. The arrival rate of source $i$ at any server is $ \lambda_{j}^{(i)}= \lambda^{(i)} $, for all $i~\in~[m], j~\in ~[n]$. The arrival rate of the sources other than source $i$ is $\overline{\lambda^{(i)}} ~\triangleq~\sum_{i' \neq i} \lambda^{(i')},  i~  \in~[m]$. The service rate at any server is $\mu$. Let $\Delta_{i}$ denote the average AoI at the monitor for source $i~  \in~[m]$. Without loss of generality, we calculate $\Delta_1$ for source $1$. In the queue model, upon arrival of a new update from any source, each server immediately drops any previous update in service and starts to serve the new update.  

The continuous state $\mathbf{x}$ represents the age for source $1$, and similar to the single-source case, it is defined as follows:
$x_0$ is AoI of source $1$ at the monitor, $x_1$ is the age of the freshest update among all updates of source $1$ in the servers, $x_2$ corresponds to the second freshest update in the servers, etc. Therefore $x_{1} \leq x_{2} \leq .... \leq x_{n}$, for any time. 
Using fake updates and fake preemption as explained in Section~\ref{LCFS}, we obtain an SHS with a single discrete state and $3n$ transitions described below: 
\begin{comment}
\begin{itemize}
\item	$0\le l \le n-1$: A fresh update arrives at virtual server $1 \le l+1 \le n$ from source 1: this fresh update becomes the freshest update: so $x_1' = 0$. Now, the old freshest update becomes the second freshest update, that is $x_2' = x_1$, and so on. Then $\mathbf{x}' = [x_0,0, x_1, \ldots,x_{l}, x_{l+2}, \ldots,x_{n } ]$. The transition rate is~$\lambda^{(1)}$.

\item	$n\le l \le 2n-1$: A fresh update arrives at virtual server $1 \le l+1-n \triangleq l' \le n$ from source $i \neq 1$: The age at the source does not change $x_0' = x_0$. The $l'$-th freshest update is lost. Moreover, if the virtual server $l'$ does complete service, it does not reduce the age of the process of interest. Thus, the $l'$-th virtual server becomes the $n$-th virtual server. Therefore, we have
$\mathbf{x}' = [x_0,x_1,   \ldots, x_{l' -1}, x_{l' +1} \ldots , x_{n}, x_0]$.  The transition rate is $\overline{\lambda^{(1)}}$.

\item $2 n\le l \le 3n-1$: the update in the $l+1-2n \triangleq h $ is delivered. The age $x_0$ is reset to $x_{h}$ and the virtual server $h$ becomes idle. using fake update and fake preemption, we reset $x_l' = x_h, h \le j \le  n$. The transition rate is $\mu$.

\end{itemize}
\end{comment}

{\bf Case I.}	$l \in [0: n-1]$: A fresh update arrives at virtual server $l$ from source 1. This update is the freshest update, so $x_1' = 0$. Now, the previous freshest update becomes the second freshest update, that is $x_2' = x_1$, and so on. Then $\mathbf{x}' = [x_0,0, x_1, \ldots,x_{l}, x_{l+2}, \ldots,x_{n } ]$. The transition rate is~$\lambda^{(1)}$.

{\bf Case II.}	$l \in [n: 2n-1]$: A fresh update arrives at virtual server $l' \triangleq l+1-n$ from source $i \neq 1$. The age at the monitor does not change, namely, $x_0' = x_0$. The $l'$-th freshest update is preempted. Moreover, if the virtual server $l'$ does complete service, it does not reduce the age of the source of interest. Thus, the $l'$-th virtual server becomes the $n$-th virtual server with age $x_0$. Therefore, we have
$\mathbf{x}' = [x_0,x_1,   \ldots, x_{l' -1}, x_{l' +1} \ldots , x_{n}, x_0]$. The transition rate is $\overline{\lambda^{(1)}}$.

{\bf Case III.} $l \in [2n: 3n-1]$: the update of source $1$ in virtual server $h \triangleq l+1-2n $ is delivered. The age $x_0$ is reset to $x_{h}$ and the virtual server $h$ becomes idle. Using fake update and fake preemption, we reset $x_l' = x_h, h \le j \le  n$. The transition rate is $\mu$.

Dropping the index $q=0$ and denoting $\mathbf{v}_0=\mathbf{v}=[v_0,v_1,\dots,v_n]$, the system of equations for the model is
%The system of equations for aforementioned model becomes: 
%\vspace{-0.35cm}
%\begin{small}
%\begin{align*} 
%%\label{multi-source}
%(n \lambda_1 + n \lambda_2 + n \mu) [v_0, v_1,\ldots, v_n]&=  
%[1,1,1,\ldots,1,1,1,1]  \nonumber \\
%&+ \lambda_1 [v_0, 0, v_2, v_3, \ldots, v_n] \nonumber\\
%&+ \lambda_1 [v_0, 0, v_1, v_3, \ldots, v_n] \nonumber\\
%&+ \lambda_1 [v_0, 0, v_1, v_2, \ldots, v_n] \nonumber\\
%  &  \qquad \vdots \nonumber \\
%&+ \lambda_1 [v_0, 0, v_1, v_2, \ldots, v_{n-1}] \nonumber\\
%&+ \lambda_2 [v_0, v_2, v_3,   \ldots, v_{n}, v_0]\nonumber\\
%&+ \lambda_2 [v_0, v_1, v_3,   \ldots, v_{n}, v_0]\nonumber\\
%&+ \lambda_2 [v_0, v_1, v_2,   \ldots, v_{n}, v_0]\nonumber\\
%& \qquad \vdots \nonumber \\
%&+ \lambda_2  [v_0, v_1, v_2,   \ldots, v_{n-1}, v_0]\nonumber\\
%&+ \mu [v_1, v_1, v_1, v_1, \ldots, v_{1}] \nonumber\\
%&+ \mu [v_2, v_1, v_2, v_2, \ldots, v_{2}]\nonumber\\
%&+ \mu [v_3, v_1, v_2, v_3, \ldots, v_{3}]\nonumber\\
%& \qquad \vdots  \nonumber\\
%&+ \mu [v_n, v_1, v_2, v_3, \ldots, v_{n}].
%\end{align*}
%\end{small}
%\vspace{-0.35cm}
%The above system of equations is equivalent to 
\begin{align}
n \mu v_{0} &= 1 + \mu \sum_{i=1}^{n} v_{i} , \nonumber\\
v_{1}(\overline{\lambda^{(1)}}+n\lambda^{(1)})&=1+ \overline{\lambda^{(1)}}v_{2}, \nonumber\\
n (\lambda + \mu ) v_i &= 1 + (i-1) \lambda^{(1)} v_i + (n-i+1) \lambda^{(1)} v_{i-1} \nonumber\\ 
&+i \overline{\lambda^{(1)}} v_{i+1} +  (n-i) \overline{\lambda^{(1)}} v_i \nonumber\\ 
&+ \mu \sum\limits_{j=1}^{i-1} v_j + (n-i+1) \mu v_i, \quad 2 \le i  \le n,
\label{equiv_system}
\end{align}
\noindent	where $v_{n+1} \triangleq v_0$ and $\lambda=\overline{\lambda^{(1)}}+\lambda^{(1)}=\sum_{i=1}^{n} \lambda_i$.

The theorems below state the average AoI for $n=2,3$ servers, and determine the optimal arrival rate given the sum arrival rate.
\begin{theorem}
\label{prop:all_2_servers}
Let $AoI_{i}$ denote AoI at the monitor for source $i$. For $m$ information sources and $n=2$ servers, we have
\begin{align}
\Delta_{i} = \frac{1}{2 (\lambda + \mu)} + \frac{\lambda + \mu}{2 \mu \lambda^{(i)}}, \quad 1 \leq i \leq m.
\label{all_2_servers}
\end{align}
\end{theorem}
\begin{proof}
From \eqref{equiv_system}, we write
\begin{comment}
\begin{align*}
v_0 &= \frac{1}{n \mu} + \frac{v_1  +v_2}{n},
v_1  = \frac{1}{\lambda + \lambda_1} + \frac{\lambda_2 v_2 }{\lambda + \lambda_1} \\
v_2 &= \frac{1}{\lambda + \lambda_2 + \mu }   + \frac{(\lambda_1 +\mu) v_1 }{\lambda + \lambda_2 + \mu  }
 + \frac{n \lambda_2 v_0 }{\lambda + \lambda_2 + \mu  }
\end{align*}
\end{comment}
\begin{align*}
n (\lambda + \mu) [v_0,v_1,v_2] &= [1,1,1] \\
+& \lambda^{(1)} [v_0, 0, v_2] \\
+& \lambda^{(1)} [v_0, 0, v_1] \\
+& \overline{\lambda^{(1)}} [v_0, v_2, v_0] \\
+& \overline{\lambda^{(1)}} [v_0, v_1, v_0] \\
+& \mu [v_1, v_1, v_1] \\
+& \mu [v_2, v_1, v_2]
\end{align*}
From the $0$-th coordinate, we have
\begin{align*}
n (\lambda + \mu)v_0 &= 1 + n \lambda v_0 +  \mu (v_1  +v_2)\\
n \mu v_0  &= 1 +  \mu (v_1  +v_2) \\
v_0 &= \frac{1}{n \mu} + \frac{v_1  +v_2}{n}
\end{align*}
From the $1$-st coordinate, we have
\begin{align*}
2 (\lambda + \mu)v_1 &= 1 + \overline{\lambda^{(1)}}  (v_1  +v_2) + 2 \mu v_1 \\
2 \lambda v_1 &= 1 + \overline{\lambda^{(1)}} v_1 + \overline{\lambda^{(1)}} v_2  \\
(\lambda + \lambda^{(1)}) v_1 &= 1 + \lambda_2 v_2 \\
v_1 &= \frac{1}{\lambda + \lambda^{(1)}} + \frac{\overline{\lambda^{(1)}} v_2 }{\lambda + \lambda^{(1)}}
\end{align*}
From the $2$-nd coordinate, we have
\begin{align*}
2 (\lambda + \mu)v_2 = 1 + \lambda^{(1)}  (v_1  +v_2) + 2 \overline{\lambda^{(1)}} v_0 +  \mu  (v_1  +v_2) \\
2 (\lambda + \mu)v_2 = 1 + (\lambda^{(1)} +\mu) v_1    + (\lambda^{(1)} +\mu  ) v_2  + 2 \overline{\lambda^{(1)}} v_0   \\
 ( \lambda + \overline{\lambda^{(1)}} + \mu  )v_2 =  1 + (\lambda^{(1)} +\mu) v_1  + n \overline{\lambda^{(1)}} v_0  \\
 v_2 = \frac{1}{\lambda + \overline{\lambda^{(1)}} + \mu }   + \frac{(\lambda^{(1)} +\mu) v_1 }{\lambda + \overline{\lambda^{(1)}} + \mu  }
 + \frac{n \overline{\lambda^{(1)}} v_0 }{\lambda + \overline{\lambda^{(1)}} + \mu  }
\end{align*}
Solving these equations followed by algebraic simplifications results in \eqref{all_2_servers}.
%, we get $v_0$ to be 
%\begin{align*}
%v_0  = \frac{1}{2 (\lambda + \mu)} + \frac{\lambda + \mu}{2 \mu \lambda_1} 
%= \frac{1}{2 \mu} ( \frac{1}{1 +\rho} + \frac{1 + \rho}{\rho_1}	).
%\end{align*}
\end{proof}

\begin{theorem}
\label{prop:all_3_servers}
For $m$ information sources and $n=3$ servers, we have
\begin{align*}
\Delta_i  = \frac{1}{3 \mu} \frac{(5 \rho^{(1)} + 2 (\rho + 1)^2) (\rho + 1)}{2 \rho^3 + 5 \rho^{(1)} \rho + 2 \rho^{(1)}}, \quad  1 \leq i  \leq m,  
\end{align*}
where $\rho= \frac{\lambda}{\mu}$ and $\rho^{(i)}= \frac{\lambda^{(i)}}{\mu}$.
\end{theorem}
\begin{proof}
\begin{align*}
n (\lambda + \mu) [v_0,v_1,v_2,v_3]  = &[1,1,1,1] \\
+ \lambda^{(1)} & [v_0, 0, v_2, v_3] \\
+ \lambda^{(1)} &[v_0, 0, v_1, v_3] \\
+ \lambda^{(1)} &[v_0, 0, v_1, v_2] \\
+ \overline{\lambda^{(1)}} & [v_0, v_2, v_3, v_0] \\
+ \overline{\lambda^{(1)}} &[v_0, v_1, v_3,v_0] \\
+ \overline{\lambda^{(1)}} &[v_0, v_1, v_2, v_0 ] \\
+ \mu &[v_1, v_1, v_1, v_1] \\
+ \mu &[v_2, v_1, v_2, v_2] \\
+ \mu &[v_3, v_1, v_2, v_3].
\end{align*}
At the $0$-th coordinate
\begin{align*}
n (\lambda + \mu)  v_0 &=  1 + n\lambda  v_0 + \mu (v_1 +v_2 +v_3) \\
n  \mu v_0 &= 1 + \mu(v_1 +v_2 +v_3) \\
v_0 &= \frac{1}{n  \mu} +\frac{v_1 +v_2 +v_3}{3}
\end{align*}
At the $1$-st coordinate
\begin{align*}
3 (\lambda + \mu)  v_1 &=  1 +  \overline{\lambda^{(1)}}  v_2 + 2  \overline{\lambda^{(1)}}  v_1 + 3 \mu v_1 \\
3 \lambda v_1 &=  1 +  \overline{\lambda^{(1)}}  v_2 + 2  \overline{\lambda^{(1)}}  v_1  \\
(\lambda + 2 \lambda^{(1)} )v_1 &= 1 +  \overline{\lambda^{(1)}}  v_2 \\
v_1 &= \frac{1}{\lambda + 2 \lambda^{(1)}} + \frac{\overline{\lambda^{(1)}}}{\lambda + 2 \lambda^{(1)}}v_2
\end{align*}
At the $2$-nd coordinate
\begin{align*}
3 (\lambda + \mu)  v_2 =  1 +  \lambda^{(1)}  v_2 +  2 \lambda^{(1)} v_1 + 2 \overline{\lambda^{(1)}} v_3\\ + \overline{\lambda^{(1)}} v_2 + \mu v_1 + 2 \mu v_2,\\
(3\lambda +  3 \mu  -\lambda^{(1)} -\overline{\lambda^{(1)}}-  2 \mu) v_2 = 1+  \\ (2 \lambda^{(1)} + \mu) v_1 + 2 \overline{\lambda^{(1)}} v_3, \\
( 2 \lambda +  \mu) v_2 = 1   + (2 \lambda^{(1)} + \mu) v_1 + 2 \overline{\lambda^{(1)}} v_3 \\
\end{align*}
At the $3$-rd coordinate
\begin{align*}
 (3 \lambda + 3 \mu)  v_3 =  1 +   2 \lambda^{(1)} v_3 + \lambda^{(1)}  v_2 +\\ 3 \overline{\lambda^{(1)}} v_0 + \mu (v_1 + v_2 +v_3), \\
  (3 \lambda + 3 \mu - 2 \lambda^{(1)} - \mu)  v_3 =  1 + \mu v_1 + (\lambda^{(1)} + \mu) v_2 \\
  (\lambda + 2 \overline{\lambda^{(1)}}  + 2 \mu)v_3 =  1 + \mu v_1 + (\lambda^{(1)} + \mu) v_2 
\end{align*}

Then, we have 
\begin{align}
v_0 = \frac{(\lambda + \mu) (2 \lambda^2 + 4 \lambda \mu + 2 \mu^2 + 5 \lambda^{(1)} \mu)}{6 \lambda^3 \mu + 15 \lambda^{(1)} \lambda \mu^2 + 6 \lambda^{(1)} \mu^3},
\end{align}

And the age is 
 \begin{align*}
 v_0 &= \frac{(\lambda + \mu)  ( 2   (\lambda+ \mu)^2 + 5 \lambda^{(1)} \mu)}{3 \mu (2 \lambda^3 + 5   \lambda \lambda^{(1)}  \mu + 2 \lambda^{(1)} \mu^2 ) } \\
 &= \frac{(\lambda + \mu)  ( 2   (\lambda+ \mu)^2 + 5 \lambda^{(1)} \mu)}{3 \mu (2 \lambda^{(1)} (\lambda + \mu)^2 +  2 \lambda^2 \overline{\lambda^{(1)}}  + \lambda \lambda^{(1)} \mu )} \\
 &= \frac{1}{3 \mu} \frac{(5 \rho^{(1)} + 2 (\rho + 1)^2) (\rho + 1)}{2 \rho^3 + 5 \rho^{(1)} \rho + 2 \rho^{(1)}}.
 \end{align*}
\end{proof}

\begin{theorem} \label{convex}
Consider $m$ information sources and $n = 2$ servers. The optimal arrival rate ${\lambda^{(i)}}^*$ minimizing the weighted sum of AoIs in theorem \ref{prop:all_2_servers}, i.e., $w_1 \Delta_1 + w_2 \Delta_2+...+w_n \Delta_n$ for $w_i \geq 0$, subject to the constraint $\lambda^{(1)}+ \lambda^{(2)}+...+ \lambda^{(m)} = \lambda$, is given by
\begin{align*}
{\lambda^{(i)}}^{*}=\frac{\lambda \sqrt{w_i} }{\sum_{i=1}^{m} \sqrt{w_i}},  i \in [m]. 
\end{align*}
\end{theorem}
\begin{proof}
The objective function that we are trying to minimize is convex (it is obvious from the second derivative matrix) and therefore we just have to put the derivative with respect to each $\lambda^{(i)}$ equal to zero.
\begin{align} \label{deri}
 \frac{\partial}{\partial \lambda^{(i)}}   ( w_1 \Delta_1 + w_2 \Delta_2+...+w_n \Delta_n +a(\sum_{i=1}^{n} \lambda^{(i)}-\lambda)) =0,
\end{align}
for i $\in [m]$. Simplifying \eqref{deri} results in:
\begin{align}
    \frac{w_{1}}{(\lambda^{(1)})^{2}}= \frac{w_{2}}{(\lambda^{(2)})^{2}}=\dots=\frac{w_{n}}{(\lambda^{(n)})^{2}}=a.
\end{align}
Knowing the fact that $\lambda^{(1)}+ \lambda^{(2)}+...+ \lambda^{(m)} = \lambda$, we obtain the result in theorem~\ref{convex}.
\end{proof}

\section{Heterogeneous Network for a Single Source}
\label{hetro-sec}
In this section, we consider a single source and assume that the arrival and service rates of the servers are arbitrary. 
We denote by $\lambda_{j}^{(1)} \triangleq \lambda_j$ the arrival rate of the single source at server $j$, and $\mu_j$ the service rate of server $j \in [n]$.
For this setting, we can no longer use the same technique used in the homogeneous case to reduce the state space and derive AoI. In particular, we need to keep track of the age of updates at the physical servers as well as their ordering, resulting in $n!$ number of states. In the following, we illustrate the steps for deriving AoI in the case of $n=2,3$ servers.
\begin{theorem}
Consider $m=1$ source and $n=2$ heterogeneous servers. The AoI is given by 
\begin{small}
\begin{align}
\label{single_soure_2_hetereg_servers}
&\Delta=  \\
&\frac{1}{\mu_{1}+\mu_{2}} + \frac{1}{\lambda_{1}+\lambda_{2}}+\frac{1}{\mu_{1}+\mu_{2}} \frac{1}{\lambda_{1}+\lambda_{2}} (\frac{\mu_{1}\lambda_{2}}{\lambda_{1}+\mu_{2}} + \frac{\mu_{2}\lambda_{1}}{\lambda_{2}+\mu_{1}}).\nonumber
\end{align}
\end{small}
\end{theorem}
\begin{proof}
We define state $1$ as the state that server $1$ contains a fresher update compared to server $2$ and state $2$ as the state that server $2$ has the fresher update. Upon arrival of an  update at each server or receipt of an update at the monitor, we observe some self-transition and intra-state transitions. Transitions rate and mappings are illustrated in Table~\ref{table2}.
\vspace{0.25cm}
\begin{table}
\centering
\begin{tabular}{ cccccccccc }
%\caption{Transition Table when $n=2$ heterogeneous servers}
 $l$ & $\lambda{(l)}$ & Transition&$\mathbf{x}^\prime$ =$\mathbf{x}A_{l}$ & $v_{q_{l}} A_{l}$&\\          \hline
 $1$ & $\lambda_{1}$ & $1 \rightarrow 1 $ & $[x_{0},0,x_{2}]$ &$[v_{10},0,v_{12}]$&\\ \hline
 $2$ & $\lambda_{1}$ & $2 \rightarrow 1 $ & $[x_{0},0,x_{2}]$ &$[v_{20},0,v_{22}]$&\\ \hline
 $3$ & $\lambda_{2}$ & $1 \rightarrow 2 $&$[x_{0},x_{1},0]$ &$[v_{10},v_{11},0]$&\\ \hline
 $4$ & $\lambda_{2}$ & $2 \rightarrow 2 $&$[x_{0},x_{1},0]$ &$[v_{20},v_{21},0]$&\\ \hline
 $5$ & $\mu_{1}$ & $1 \rightarrow 1 $ & $[x_{1},x_{1},x_{1}]$ &$[v_{11},v_{11},v_{11}]$&\\  \hline
 $6$ & $\mu_{1}$ & $2 \rightarrow 2 $&$[x_{1},x_{1},x_{2}]$ &$[v_{21},v_{21},v_{22}]$&\\  \hline
 $7$ & $\mu_{2}$ & $1 \rightarrow 1 $&$[x_{2},x_{1},x_{2}]$ &$[v_{12},v_{11},v_{12}]$&\\  \hline
 $8$ & $\mu_{2}$ & $2 \rightarrow 2 $&$[x_{2},x_{2},x_{2}]$ &$[v_{22},v_{22},v_{22}]$& \\  
\end{tabular}
\caption{Table of transitions for $n=2$ heterogeneous servers.}
\label{table2}
\end{table}
Steady states probabilities are found knowing that $\pi_{1}+\pi_{2}=1$ and $\pi_{1} \lambda_{2}= \pi_{2} \lambda_{1}$. Therefore, we will have $\mathbf{\pi}=[\frac{\lambda_{1}}{\lambda_{1}+\lambda_{2}}, \frac{\lambda_{2}}{\lambda_{1}+\lambda_{2}}]$. 
\begin{multline}
(\lambda_{1}+\lambda_{2}+\mu_{1}+\mu_{2})\mathbf{v_{1}}= b_{1} \pi_{1}+
\lambda_{1}(v_{10},0,v_{12})+\lambda_{1}(v_{20},0,v_{22}) \\+
\mu_{1}(v_{11},v_{11},v_{11})+\mu_{2}(v_{12},v_{11},v_{12})
\end{multline}
\begin{multline}
(\lambda_{1}+\lambda_{2}+\mu_{1}+\mu_{2})\mathbf{v_{2}}= b_{2} \pi_{2}+
\lambda_{2}(v_{10},v_{11},0)+\lambda_{2}(v_{20},v_{21},0) \\+
\mu_{1}(v_{21},v_{21},v_{22})+\mu_{2}(v_{22},v_{22},v_{22})
\end{multline}

Where $\mathbf{v_{1}}=(v_{10},v_{11},v_{12})$ and $\mathbf{v_{2}}=(v_{20},v_{21},v_{22})$. Therefore, we have six equations and six unknowns here. We can easily see that $v_{11}=\frac{\pi_{1}}{\lambda_{1}+\lambda_{2}}$ and $v_{22}=\frac{\pi_{2}}{\lambda_{1}+\lambda_{2}}$.
\begin{align*}
v_{12}&=  \frac{\pi_{1}}{\lambda_{2}+\mu_{1}}+\frac{\lambda_{1}\pi_{2}}{(\lambda_{1}+\lambda_{2})(\lambda_{2}+\mu_{1})}+\frac{\mu_{1}\pi_{1}}{(\lambda_{1}+\lambda_{2})(\lambda_{2}+\mu_{1})} \\
&=  \pi_{1} (\frac{1}{\lambda_{1}+\lambda_{2}} + \frac{1}{\lambda_{2}+\mu_{1}})
\end{align*}

\begin{align*}
v_{21}&=\frac{\pi_{2}}{\lambda_{1}+\mu_{2}}+\frac{\lambda_{2}\pi_{1}}{(\lambda_{1}+\lambda_{2})(\lambda_{1}+\mu_{2})}+\frac{\mu_{2}\pi_{2}}{(\lambda_{1}+\lambda_{2})(\lambda_{1}+\mu_{2})} \\
&=  \pi_{2} (\frac{1}{\lambda_{1}+\lambda_{2}} + \frac{1}{\lambda_{1}+\mu_{2}})
\end{align*}

\begin{equation*}
(\lambda_{1}+\lambda_{2}+\mu_{1}+\mu_{2})v_{10}=\pi_{1}+\lambda_{1}v_{10}+\lambda_{1}v_{20}+\mu_{1}v_{11}+\mu_{2}v_{12}
\end{equation*}
\begin{equation*}
(\lambda_{1}+\lambda_{2}+\mu_{1}+\mu_{2})v_{20}=\pi_{2}+\lambda_{2}v_{10}+\lambda_{2}v_{20}+\mu_{1}v_{21}+\mu_{2}v_{22}
\end{equation*}
We add this 2 equations together and simplify it. Age of Information at the monitor is equal to $v_{10}+v_{20}$ which is:
\begin{align*}
& AoI= \frac{1}{\mu_{1}+\mu_{2}}+ \frac{\mu_{1}(v_{11}+v_{21})+\mu_{2}(v_{12}+v_{22})}{\mu_{1}+\mu_{2}}= \\
&\frac{1}{\mu_{1}+\mu_{2}} + \frac{1}{\lambda_{1}+\lambda_{2}}+\frac{1}{\mu_{1}+\mu_{2}} \frac{1}{\lambda_{1}+\lambda_{2}} (\frac{\mu_{1}\lambda_{2}}{\lambda_{1}+\mu_{2}} + \frac{\mu_{2}\lambda_{1}}{\lambda_{2}+\mu_{1}})
\end{align*}

\end{proof}
Next, for $n=2$ servers, we find the optimal arrival rates of servers, ${\lambda_1}^{*},{\lambda_2}^{*}$,  given fixed service rates $\mu_1,\mu_2$ and sum arrival rate $\lambda \triangleq \lambda_1+\lambda_2$. The optimal ${\lambda_1}^{*}$ is illustrated in Figure~\ref{optimal_lambda}.
\begin{figure}
\centering
\includegraphics[width = 0.4\textwidth]{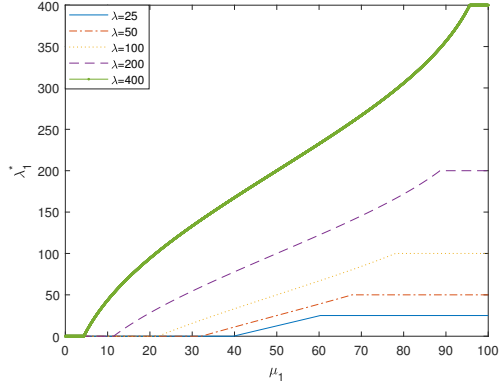} 
\caption{Optimal value of $\lambda_1$ as a function of $\mu_1$. $\lambda_1+\lambda_2=\lambda, \mu_1+\mu_2=100$.}
\label{optimal_lambda}
\end{figure}

\begin{theorem}
\label{thm:optimal_lambda_heto_n_2}
For $m=1$ and $n=2$ heterogeneous servers,  given $\mu_1,\mu_2$ and fixed $\lambda_1+\lambda_2=\lambda$, the optimal ${\lambda_{1}}^{*}$ satisfies

\noindent $\bullet$ if  $\mu_{1} < \mu_{2}$ and $\mu_{2}^2 - \frac{\mu_{1}(\lambda+\mu_{1})(\lambda+\mu_{2})}{\mu_{2}} < 0$,
\begin{align*}
 {\lambda_{1}}^{*}=\frac{-(\mu_{2}+c(\lambda+\mu_{1}))+\sqrt{\mu_{1}(\lambda+\mu_{2})(2+\frac{\mu_{2}}{\lambda+\mu_{1}}+\frac{\lambda+\mu_{1}}{\mu_{2}})}}{1- \frac{\mu_{1}(\lambda+\mu_{2})}{\mu_{2}(\lambda+\mu_{1})}},  
\end{align*}
 \noindent $\bullet$  if  $\mu_{1} < \mu_{2}$ and $\mu_{2}^2 - \frac{\mu_{1}(\lambda+\mu_{1})(\lambda+\mu_{2})}{\mu_{2}} \geq 0: $ $  {\lambda_1}^{*}=0, {\lambda_2}^{*}=\lambda,$
     %\begin{align*}
    %    \lambda_1^{*}=0, \lambda_2^{*}=\lambda,
  %  \end{align*}
\noindent $\bullet$ if $\mu_{1}>\mu_{2}$ and $\mu_{1}^2 \geq \frac{\mu_{2}(\lambda+\mu_{1})(\lambda+\mu_{2})}{\mu_{1}}:$ ${\lambda_{1}}^{*}=\lambda, {\lambda_{2}}^{*}=0,$
 %  \begin{align*}
   %      \lambda_{1}^{*}=\lambda, \lambda_{2}^{*}=0,
 %  \end{align*}
 
 \noindent $\bullet$    if $\mu_{1}>\mu_{2}$ and $\mu_{1}^2 < \frac{\mu_{2}(\lambda+\mu_{1})(\lambda+\mu_{2})}{\mu_{1}}$.
     \begin{align*}
{\lambda_{1}}^{*}= \lambda-\frac{-(\mu_{1}+\frac{(\lambda+\mu_{2})}{c})+\sqrt{\mu_{2}(\lambda+\mu_{1})(2+\frac{\mu_{1}}{\lambda+\mu_{2}}+\frac{\lambda+\mu_{2}}{\mu_{1}})}}{1- \frac{\mu_{2}(\lambda+\mu_{1})}{\mu_{1}(\lambda+\mu_{2})}},
\end{align*}  
%\begin{small}
%\begin{align*}
 %   \lambda_1^{*}=\frac{-(\mu_{2}+c(\lambda+\mu_{1}))+\sqrt{\mu_{1}(\lambda+\mu_{2})(2+\frac{\mu_{2}}{\lambda+\mu_{1}}+\frac{\lambda+\mu_{1}}{\mu_{2}})}}{1- \frac{\mu_{1}(\lambda+\mu_{2})}{\mu_{2}(\lambda+\mu_{1})}}, \\
 %   \end{align*}
  %    \vspace{-0.2cm}
    %if $\mu_{1} < \mu_{2}$ and $\mu_{2}^2 - \frac{\mu_{1}(\lambda+\mu_{1})(\lambda+\mu_{2})}{\mu_{2}} \geq 0$.
%    \begin{align*}
  %      \lambda_1^{*}=0, \lambda_2^{*}=\lambda,
 %   \end{align*}
%    \vspace{-0.2cm}
%   if  $\mu_{1} < \mu_{2}$ and $\mu_{2}^2 - \frac{\mu_{1}(\lambda+\mu_{1})(\lambda+\mu_{2})}{\mu_{2}} \geq 0.$
 %  \begin{align*}
 %        \lambda_{1}^{*}=\lambda, \lambda_{2}^{*}=0,
%   \end{align*}
   % \vspace{-0.2cm}
%   if $\mu_{1}>\mu_{2}$ and $\mu_{1}^2 \geq \frac{\mu_{2}(\lambda+\mu_{1})(\lambda+\mu_{2})}{\mu_{1}}$.
 %  \begin{align*}
%\lambda_{1}^{*}= \lambda-\frac{-(\mu_{1}+\frac{(\lambda+\mu_{2})}{c})+\sqrt{\mu_{2}(\lambda+\mu_{1})(2+\frac{\mu_{1}}{\lambda+\mu_{2}}+\frac{\lambda+\mu_{2}}{\mu_{1}})}}{1- \frac{\mu_{2}(\lambda+\mu_{1})}{\mu_{1}(\lambda+\mu_{2})}},
%\end{align*}
  %    \vspace{-0.2cm}
   % if $\mu_{1}>\mu_{2}$ and $\mu_{1}^2 < \frac{\mu_{2}(\lambda+\mu_{1})(\lambda+\mu_{2})}{\mu_{1}}$.
%\end{small}
\end{theorem}

\noindent where $c= \frac{\mu_{1}(\lambda+\mu_{2})}{\mu_{2}(\lambda+\mu_{1})}$.
\begin{proof}
In order to find the optimal values of $\lambda_{1}$ and $\lambda_{2}$ for a given values of $\mu_1,\mu_{2}, \lambda$ where $\lambda_{1}+\lambda_{2}=\lambda$, we set the derivative of the following equation with respect to $\lambda_{1}$, $\lambda_{2}$ and $a$ to zero.
\begin{align*}
\frac{1}{\mu_{1}+\mu_{2}}+ \frac{\mu_{1}(v_{11}+v_{21})+\mu_{2}(v_{12}+v_{22})}{\mu_{1}+\mu_{2}} - a(\lambda_{1}+\lambda_{2}-\lambda)
\end{align*}
\begin{align*}
\frac{\partial AoI}{\partial \lambda_{1}} = \frac{-1}{(\lambda_{1}+\lambda_{2})^2} 
-\frac{\mu_{1}\lambda_{2}(2\lambda_{1}+\lambda_{2}+\mu_{2})}{(\lambda_{1}+\lambda_{2})^2 (\lambda_{1}+\mu_{2})^2}\\+\frac{(\lambda_{2}+\mu_{1})(\mu_{2}\lambda_{2})}{(\lambda_{1}+\lambda_{2})^2 (\lambda_{2}+\mu_{1})^2}-a =0
\end{align*} 

\begin{align*}
\frac{\partial AoI}{\partial \lambda_{2}} = \frac{-1}{(\lambda_{1}+\lambda_{2})^2} 
-\frac{\mu_{2}\lambda_{1}(2\lambda_{2}+\lambda_{1}+\mu_{1})}{(\lambda_{1}+\lambda_{2})^2 (\lambda_{2}+\mu_{1})^2}\\+\frac{(\lambda_{1}+\mu_{2})(\mu_{1}\lambda_{1})}{(\lambda_{1}+\lambda_{2})^2 (\lambda_{1}+\mu_{2})^2}-a =0
\end{align*} 
Also, we know that $\lambda_{1}+\lambda_{2}=\lambda$. With some algebraic simplification we reach to this $2nd$ order polynomial in order to find the optimal value of $\lambda_{1}$ and consequently $\lambda_{2}$.
\begin{align} \label{opti}
\lambda_{1}^2 (1-c)+ 2\lambda_{1} (\mu_{2}+c(\lambda+\mu_{1})) + \mu_{2}^2 - c(\lambda+\mu_{1})^2,
\end{align}

\noindent where $c= \frac{\mu_{1}(\lambda+\mu_{2})}{\mu_{2}(\lambda+\mu_{1})}$.

When $c=1$ it is equivalent to $\mu_{1}=\mu_{2}$ and the equation \ref{opti} becomes a first order polynomial which results in $\lambda_{1}=\lambda_{2}= \frac{\lambda}{2}$.
This polynomial has $2$ real roots because of its positive discriminant and therefore solving the equation \ref{opti} gives us $2$ possible candidate for our optimization problem. When $\mu_{1} < \mu_{2}$ then $c<1$. Knowing the fact that for $2$ roots of \ref{opti} we have,

\begin{align*}
r_{1}+r_{2} = \frac{\mu_2 + \frac{\mu_{1}(\lambda+\mu_{2})}{\mu_{2}}}{c-1}, \\
r_{1}r_{2} = \frac{\mu_{2}^2 - \frac{\mu_{1}(\lambda+\mu_{1})(\lambda+\mu_{2})}{\mu_{2}}}{1-c}.
\end{align*}

As a result, when $\mu_{1} < \mu_{2}$ and $\mu_{2}^2 - \frac{\mu_{1}(\lambda+\mu_{1})(\lambda+\mu_{2})}{\mu_{2}} \geq 0$, the $2$ roots are negative and therefore in this regime our optimal values become $\lambda_{1}=0, \lambda_{2}=\lambda$. When $\mu_{1} < \mu_{2}$ and $\mu_{2}^2 - \frac{\mu_{1}(\lambda+\mu_{1})(\lambda+\mu_{2})}{\mu_{2}} \geq 0$, the positive root is the optimal rate which is equal to:
\begin{align*}
\lambda_{1}=\frac{-(\mu_{2}+c(\lambda+\mu_{1}))+\sqrt{\mu_{1}(\lambda+\mu_{2})(2+\frac{\mu_{2}}{\lambda+\mu_{1}}+\frac{\lambda+\mu_{1}}{\mu_{2}})}}{1- \frac{\mu_{1}(\lambda+\mu_{2})}{\mu_{2}(\lambda+\mu_{1})}}.
\end{align*}

Similarly by writing the $2-nd$ order polynomial for $\lambda_{2}$, we reach to the conclusion that when $\mu_{1}>\mu_{2}$ , if $\mu_{1}^2 \geq \frac{\mu_{2}(\lambda+\mu_{1})(\lambda+\mu_{2})}{\mu_{1}}$ the optimal rates are $\lambda_{1}=\lambda, \lambda_{2}=0$.
In the regime that $\mu_{1}>\mu_{2}$ and $\mu_{1}^2 < \frac{\mu_{2}(\lambda+\mu_{1})(\lambda+\mu_{2})}{\mu_{1}}$, the positive root is the optimal rate.
\begin{align*}
\lambda_{2}= \frac{-(\mu_{1}+\frac{(\lambda+\mu_{2})}{c})+\sqrt{\mu_{2}(\lambda+\mu_{1})(2+\frac{\mu_{1}}{\lambda+\mu_{2}}+\frac{\lambda+\mu_{2}}{\mu_{1}})}}{1- \frac{\mu_{2}(\lambda+\mu_{1})}{\mu_{1}(\lambda+\mu_{2})}}.
\end{align*}
\end{proof}
\vspace{0.1cm}
When $\mu_1=\mu_2$ the optimal rates that minimize AoI are ${\lambda_{1}}^{*}={\lambda_{2}}^{*}=\frac{\lambda}{2}$. As Figure \ref{optimal_lambda} illustrates, for $\mu_1=\mu_2=50$, optimal rates are ${\lambda_{1}}^{*}=\frac{\lambda}{2}$ and in the regimes that one of the service rates is much greater than the other one, AoI minimizes when all the updates are sent to the server with greater service rate.  

\begin{theorem}
Consider $m=1$ source and $n=3$ heterogeneous servers. The AoI is given by 
\begin{align}
\label{1_source_3_hetereg_servers}
AoI=\frac{1}{\sum_{i=1}^{3} \mu_{i}} + \frac{\mu_{1}\sum_{j=1}^{6} v_{j1}+\mu_{2}\sum_{j=1}^{6} v_{j2}+\mu_{3}\sum_{j=1}^{6} v_{j3}}{\sum_{i=1}^{3} \mu_{i}}
\end{align}
\end{theorem}
\begin{proof}
In this case, we'll have $6$ states and have to solve $24$ equations  to find $24$ unknowns. It seems quite troublesome to solve $24$ equations, however it seems to be quite straightforward. We have $18$ transitions for arrivals(we have $6$ states and $3$ servers) and $18$ transitions for arrivals at the monitor(depending on update coming from which server and which state it was before). For the sake of simplicity we define our states as $(1,2,3)$ equal to state $1$ which means update in server $1$ is the freshest update and update in server $2$ is the second freshet update. $(1,3,2)$ equal to state $2$ indicating that age of information in each server follows the order of $x_{1}\leq x_{3}\leq x_{2}$. $(2,1,3)$ equal to state $3$, $(2,3,1)$ state $4$, $(3,1,2)$ state $5$, and $(3,2,1)$ state $6$. After each transitions $\mathbf{x}$ changes to $\mathbf{x^\prime}$. When a new update arrives at server $i$ with rate $\lambda_{i}$, $x_{0}$ remains the same no matter which state we were at, $x_{i}$ becomes zero and the rest remains unchanged. For example if $i=1$, $x_{2}$ and $x_{3}$ remain unchanged. When an update arrives at the monitor from server $i$ with rate $\mu_{i}$, $x_{0}$ takes the value of age ($x_{i}$) of the received update. All the servers that have larger "age" will be preempted and on the vector $\mathbf{x^\prime}$ their value becomes equal to $x_{i}$ which is equivalent to consider a fake update inserted into to those servers after preemption. First we need to calculate our steady state probabilities. After writing down the equations and some simplifications, we will reach to the following equations:
\begin{align*}
\pi_{1} (\lambda_{2}+\lambda_{3}) = \lambda_{1} (\pi_{3}+\pi_{4}) \\
\pi_{2} (\lambda_{2}+\lambda_{3}) = \lambda_{1} (\pi_{5}+\pi_{6}) \\
\pi_{3} (\lambda_{1}+\lambda_{3}) = \lambda_{2} (\pi_{1}+\pi_{2}) \\
\pi_{4} (\lambda_{1}+\lambda_{3}) = \lambda_{2} (\pi_{5}+\pi_{6}) \\
\pi_{5} (\lambda_{1}+\lambda_{2}) = \lambda_{3} (\pi_{1}+\pi_{2}) \\
\pi_{6} (\lambda_{1}+\lambda_{2}) = \lambda_{3} (\pi_{3}+\pi_{4}) \\
\end{align*}

Knowing that $\sum_{i=1}^{6} \pi_{i} =1$, we can find all the steady states probabilities.
\begin{align*}
\pi_{1}=\frac{\lambda_{1}}{\lambda_{2}+\lambda_{3}} \frac{\lambda_{2}}{\lambda_{1}+\lambda_{2}+\lambda_{3}}, \\
\pi_{2}=\frac{\lambda_{1}}{\lambda_{2}+\lambda_{3}} \frac{\lambda_{3}}{\lambda_{1}+\lambda_{2}+\lambda_{3}}, \\
\pi_{3}=\frac{\lambda_{2}}{\lambda_{1}+\lambda_{3}} \frac{\lambda_{1}}{\lambda_{1}+\lambda_{2}+\lambda_{3}}, \\
\pi_{4}=\frac{\lambda_{2}}{\lambda_{1}+\lambda_{3}} \frac{\lambda_{3}}{\lambda_{1}+\lambda_{2}+\lambda_{3}}, \\
\pi_{5}=\frac{\lambda_{3}}{\lambda_{1}+\lambda_{2}} \frac{\lambda_{1}}{\lambda_{1}+\lambda_{2}+\lambda_{3}}, \\
\pi_{6}=\frac{\lambda_{3}}{\lambda_{1}+\lambda_{2}} \frac{\lambda_{2}}{\lambda_{1}+\lambda_{2}+\lambda_{3}}. \\
\end{align*}
We can easily find these $6$ parameter values.
\begin{align*}
v_{11}=\frac{\pi_{1}}{\lambda_{1}+\lambda_{2}+\lambda_{3}}, 	v_{21}= \frac{\pi_{2}}{\lambda_{1}+\lambda_{2}+\lambda_{3}} \\
v_{32}=\frac{\pi_{3}}{\lambda_{1}+\lambda_{2}+\lambda_{3}}, v_{42}=\frac{\pi_{4}}{\lambda_{1}+\lambda_{2}+\lambda_{3}} \\
v_{53}=\frac{\pi_{5}}{\lambda_{1}+\lambda_{2}+\lambda_{3}},
v_{63}=\frac{\pi_{6}}{\lambda_{1}+\lambda_{2}+\lambda_{3}}.
\end{align*}
\begin{align*}
v_{12}&=\frac{\pi_{1}+\lambda_{1}v_{32}+\lambda_{1}v_{42}+\mu_{1}v_{11}}{\lambda_{2}+\lambda_{3}+\mu_{1}}, \\
&= \pi_{1} (\frac{1}{\lambda_{1}+\lambda_{2}+\lambda_{3}} + \frac{1}{\lambda_{2}+\lambda_{3}+\mu_{1}})\\
v_{23}&=\frac{\pi_{2}+\lambda_{1}v_{53}+\lambda_{1}v_{63}+\mu_{1}v_{21}}{\lambda_{2}+\lambda_{3}+\mu_{1}}, \\
&= \pi_{2} (\frac{1}{\lambda_{1}+\lambda_{2}+\lambda_{3}} + \frac{1}{\lambda_{2}+\lambda_{3}+\mu_{1}})\\
v_{31}&=\frac{\pi_{3}+\lambda_{2}v_{11}+\lambda_{2}v_{21}+\mu_{2}v_{32}}{\lambda_{1}+\lambda_{3}+\mu_{2}}, \\
&= \pi_{3} (\frac{1}{\lambda_{1}+\lambda_{2}+\lambda_{3}} + \frac{1}{\lambda_{1}+\lambda_{3}+\mu_{2}})\\
v_{43}&=\frac{\pi_{4}+\lambda_{2}v_{53}+\lambda_{2}v_{63}+\mu_{2}v_{42}}{\lambda_{1}+\lambda_{3}+\mu_{2}}, \\
&= \pi_{4} (\frac{1}{\lambda_{1}+\lambda_{2}+\lambda_{3}} + \frac{1}{\lambda_{1}+\lambda_{3}+\mu_{2}})\\
v_{51}&=\frac{\pi_{5}+\lambda_{3}v_{11}+\lambda_{3}v_{21}+\mu_{3}v_{53}}{\lambda_{1}+\lambda_{2}+\mu_{3}}, \\
&= \pi_{5} (\frac{1}{\lambda_{1}+\lambda_{2}+\lambda_{3}} + \frac{1}{\lambda_{1}+\lambda_{2}+\mu_{3}})\\
v_{62}&=\frac{\pi_{6}+\lambda_{3}v_{32}+\lambda_{3}v_{42}+\mu_{3}v_{63}}{\lambda_{1}+\lambda_{2}+\mu_{3}}\\
&= \pi_{6} (\frac{1}{\lambda_{1}+\lambda_{2}+\lambda_{3}} + \frac{1}{\lambda_{1}+\lambda_{2}+\mu_{3}}).
\end{align*}
We calculated all the variables that are needed in this step for calculating these $6$ variables.
\begin{align*}
v_{13}=\frac{\pi_{1}+\lambda_{1}v_{33}+\lambda_{1}v_{43}+\mu_{1}v_{11}+\mu_{2}v_{12}}{\lambda_{2}+\lambda_{3}+\mu_{1}+\mu_{2}}, \\
v_{22}=\frac{\pi_{2}+\lambda_{1}v_{52}+\lambda_{1}v_{62}+\mu_{1}v_{21}+\mu_{3}v_{23}}{\lambda_{2}+\lambda_{3}+\mu_{1}+\mu_{3}},\\
v_{33}=\frac{\pi_{3}+\lambda_{2}v_{13}+\lambda_{2}v_{23}+\mu_{1}v_{31}+\mu_{2}v_{32}}{\lambda_{1}+\lambda_{3}+\mu_{1}+\mu_{2}}, \\
v_{41}=\frac{\pi_{4}+\lambda_{2}v_{51}+\lambda_{2}v_{61}+\mu_{2}v_{42}+\mu_{3}v_{43}}{\lambda_{1}+\lambda_{3}+\mu_{2}+\mu_{3}},\\
v_{52}=\frac{\pi_{5}+\lambda_{3}v_{12}+\lambda_{3}v_{22}+\mu_{1}v_{51}+\mu_{3}v_{53}}{\lambda_{1}+\lambda_{2}+\mu_{1}+\mu_{3}}, \\
v_{61}= \frac{\pi_{6}+\lambda_{3}v_{31}+\lambda_{3}v_{41}+\mu_{2}v_{62}+\mu_{3}v_{63}}{\lambda_{1}+\lambda_{2}+\mu_{2}+\mu_{3}}.
\end{align*}
Using the $12$ calculated variable before and also considering equations for $(v_{13},v_{33})$, $(v_{22},v_{52})$, and $(v_{41},v_{61})$ together we can calculate all these $6$ variables. AoI is equal to $\sum_{i=1}^{6} v_{i0}$. After 
writing equations for these $6$ variables and adding them together, we can finally calculate AoI.
\begin{align*}
AoI=\frac{1}{\sum_{i=1}^{3} \mu_{i}} + \frac{\mu_{1}\sum_{j=1}^{6} v_{j1}+\mu_{2}\sum_{j=1}^{6} v_{j2}+\mu_{3}\sum_{j=1}^{6} v_{j3}}{\sum_{i=1}^{3} \mu_{i}}
\end{align*}
\end{proof}

\section{Conclusion}
\label{conc}
In this paper, we studied the age of information in the presence of multiple independent servers monitoring several information sources. We derived AoI for the LCFS queue model using SHS analysis when we had a homogeneous network and a single source. We also provided the AoI formula for $m$ sources and $n=2,3$ servers in a homogeneous network. For a heterogeneous network, cases of $n=2,3$ servers when we have a single source were investigated and AoI formula was derived. From the simulation, it is observed that LCFS outperforms LCFS with preemption in waiting and FCFS for a homogeneous single information source network. Future directions include deriving explicit formula of AoI for multiple sources in a homogeneous and heterogeneous sensing networks where the update arrival rate and/or the service rate are different among the servers for any number of sources and servers.
\bibliographystyle{IEEEtran}
\bibliography{biblio} 
%%%%%%
%% Appendix:
%% If needed a single appendix is created by
%%
%\appendix
%%
%% If several appendices are needed, then the command
%%
% \appendices
%%
%% in combination with further \section-commands can be used.
%%%%%%

%%%%%%
%% To balance the columns at the last page of the paper use this
%% command:
%%
%\enlargethispage{-1.2cm} 
%%
%% If the balancing should occur in the middle of the references, use
%% the following trigger:
%%
\IEEEtriggeratref{3}

\end{document}